\documentclass[11pt]{article}%
\IfFileExists{sariel_computer.sty}{\def\sarielComp{1}}{}
\ifx\sarielComp\undefined%
\newcommand{\SarielComp}[1]{}
\newcommand{\NotSarielComp}[1]{#1}%
\else
\newcommand{\SarielComp}[1]{#1}%
\newcommand{\NotSarielComp}[1]{}%
\fi
\newcommand{\IfPrinterVer}[2]{#2}%

\usepackage{graphicx}
\usepackage{caption}
\usepackage[in]{fullpage}%
\usepackage{amsmath}%
\usepackage{amssymb}%
\usepackage{xcolor}%

\SarielComp{\usepackage{sariel_colors}}%

\usepackage[amsmath,thmmarks]{ntheorem}%
\theoremseparator{.}%

\usepackage{titlesec}%
\titlelabel{\thetitle. }%
\usepackage{xcolor}%
\usepackage{mleftright}%
\usepackage{xspace}%
\usepackage{hyperref}%

\IfPrinterVer{%
   \usepackage{hyperref}%
}{%
   \usepackage{hyperref}%
   \hypersetup{%
      breaklinks,%
      ocgcolorlinks, colorlinks=true,%
      urlcolor=[rgb]{0.25,0.0,0.0},%
      linkcolor=[rgb]{0.5,0.0,0.0},%
      citecolor=[rgb]{0,0.2,0.445},%
      filecolor=[rgb]{0,0,0.4},
      anchorcolor=[rgb]={0.0,0.1,0.2}%
   }
}

\theoremseparator{.}%

\theoremstyle{plain}%
\newtheorem{theorem}{Theorem}[section]

\newtheorem{lemma}[theorem]{Lemma}

\theoremstyle{plain}%
\theoremheaderfont{\sf} \theorembodyfont{\upshape}%
\newtheorem*{remark:unnumbered}[theorem]{Remark}%
\newcommand{\myqedsymbol}{\rule{2mm}{2mm}}

\theoremheaderfont{\em}%
\theorembodyfont{\upshape}%
\theoremstyle{nonumberplain}%
\theoremseparator{}%
\theoremsymbol{\myqedsymbol}%
\newtheorem{proof}{Proof:}%

\newcommand{\atgen}{\symbol{'100}}
\newcommand{\SarielThanks}[1]{\thanks{Department of Computer Science;
      University of Illinois; 201 N. Goodwin Avenue; Urbana, IL,
      61801, USA; {\tt sariel\atgen{}illinois.edu}; {\tt
         \url{http://sarielhp.org/}.} #1}}

\numberwithin{figure}{section}%
\numberwithin{table}{section}%
\numberwithin{equation}{section}%

\newcommand{\HLink}[2]{\hyperref[#2]{#1~\ref*{#2}}}
\newcommand{\HLinkSuffix}[3]{\hyperref[#2]{#1\ref*{#2}{#3}}}

\newcommand{\figlab}[1]{\label{fig:#1}}
\newcommand{\figref}[1]{\HLink{Figure}{fig:#1}}

\newcommand{\lemlab}[1]{\label{lemma:#1}}
\newcommand{\lemref}[1]{\HLink{Lemma}{lemma:#1}}%

\providecommand{\eqlab}[1]{}%
\renewcommand{\eqlab}[1]{\label{equation:#1}}
\newcommand{\Eqref}[1]{\HLinkSuffix{Eq.~(}{equation:#1}{)}}

\newcommand{\remove}[1]{}%

\newcommand{\pth}[2][\!]{\mleft({#2}\mright)}%

\renewcommand{\th}{th\xspace}
\usepackage[inline]{enumitem}

\newlist{compactenumA}{enumerate}{5}%
\setlist[compactenumA]{topsep=0pt,itemsep=-1ex,partopsep=1ex,parsep=1ex,%
   label=(\Alph*)}%

\newlist{compactenuma}{enumerate}{5}%
\setlist[compactenuma]{topsep=0pt,itemsep=-1ex,partopsep=1ex,parsep=1ex,%
   label=(\alph*)}%

\newlist{compactenumI}{enumerate}{5}%
\setlist[compactenumI]{topsep=0pt,itemsep=-1ex,partopsep=1ex,parsep=1ex,%
   label=(\Roman*)}%

\newlist{compactenumi}{enumerate}{5}%
\setlist[compactenumi]{topsep=0pt,itemsep=-1ex,partopsep=1ex,parsep=1ex,%
   label=(\roman*)}%

\newlist{compactitem}{itemize}{5}%
\setlist[compactitem]{topsep=0pt,itemsep=-1ex,partopsep=1ex,parsep=1ex,%
   label=\bullet}%

\definecolor{blue25emph}{rgb}{0, 0, 11}
\providecommand{\emphic}[2]{%
   \textcolor{blue25emph}{%
      \textbf{\emph{#1}}}%
   \index{#2}}

\providecommand{\emphi}[1]{\emphic{#1}{#1}}

\definecolor{almostblack}{rgb}{0, 0, 0.3}
\providecommand{\emphw}[1]{{\textcolor{almostblack}{\emph{#1}}}}%

\providecommand{\Mh}[1]{#1}%

\newcommand{\Term}[1]{\textsf{#1}}

\newcommand{\IncludeGraphics}[2][]{%
   \typeout{}%
   \typeout{Graphics: #2}%
   \typeout{\ includegraphics[#1]{#2}}%
   \includegraphics[#1]{#2}
   \typeout{}%
}

\newcommand{\Queue}{\Mh{\mathcal{Q}}}%

\newcommand{\obj}{\Mh{x}}%
\newcommand{\XX}{\Mh{\mathcal{X}}}%
\newcommand{\PS}{\Mh{P}}%
\newcommand{\QS}{\Mh{Q}}%

\newcommand{\ps}{\Mh{s}}%
\newcommand{\pt}{\Mh{t}}%

\newcommand{\pp}{\Mh{p}}%
\newcommand{\pu}{\Mh{u}}%
\newcommand{\pv}{\Mh{v}}%
\newcommand{\pq}{\Mh{q}}%
\newcommand{\VorC}{\Mh{\mathcal{V}}}%
\newcommand{\VorX}[1]{\Mh{\mathcal{V}}\pth{#1}}%

\newcommand{\CGAL}{\Term{CGAL}\xspace}%

\newcommand{\Cell}{\Mh{C}}%
\newcommand{\dY}[2]{\left\| #1 - #2 \right\|}%

\usepackage[geometry]{ifsym}

\newcommand{\DiskC}{\ensuremath{%
      \text{\raisebox{-2.5px}{\textcolor{gray}{\FilledCircle}}}}\xspace}%

\newcommand{\diskY}[2]{\DiskC\pth{#1, #2}}
\newcommand{\CellY}[2]{\Mh{\mathrm{cell}}_{#1}\pth{#2}}%
\newcommand{\dSetY}[2]{\Mh{\mathsf{d}}\pth{#1,#2}}%

\newcommand{\iend}{\Mh{\nu}}%

\newcommand{\RU}{\Mh{U}}%
\newcommand{\RR}{\Mh{R}}%

\usepackage{float}

\title{Shortest Secure Path in a Voronoi Diagram}%
\author{%
   Sariel Har-Peled%
   \SarielThanks{Work on this paper was partially supported by a NSF
      AF award CCF-1907400. %
   }%
   \and%
   Rajgopal Varadharajan%
}%
\date{\today}

\begin{document}

\maketitle

\begin{abstract}
    We investigate the problem of computing the shortest secure path
    in a Voronoi diagram. Here, a path is secure if it is a sequence
    of touching Voronoi cells, where each Voronoi cell in the path has
    a uniform cost of being secured. Importantly, we allow inserting
    new sites, which in some cases leads to significantly shorter
    paths. We present an $O(n \log n)$ time algorithm for solving this
    problem in the plane, which uses a dynamic additive weighted
    Voronoi diagrams.  The algorithm is a (arguably interesting)
    combination of the continuous and discrete Dijkstra algorithms.
    We also implemented the algorithm using \texttt{CGAL}.
\end{abstract}

\section{Introduction}

\paragraph{Motivation.}

Consider a facility in an environment where other facilities exist. A
client can communicate safely only with its nearest facility. The
region where one can safely communicate with the facility, is the
\emphw{Voronoi cell} of this facility in the Voronoi diagram of the
facilities \cite{bcko-cgaa-08}. Given two facilities $s$ and $t$,
consider the problem of creating a safe corridor, where one can safely
move from $s$ and $t$ while still being able to communicate safely
with both of them. This might require inserting new middle facilities,
such that the union of the new Voronoi cells (together with the
Voronoi cells of $s$ and $t$), forms a connected set (alternatively,
one can ``secure'' existing site, which has the same cost as inserting
a new one).  The natural question is where and how many sites one
needs to use/insert to establish such a reliable connection, see
\figref{two} for an example where inserting new sites dramatically
reduces the length of the path.

\begin{figure}[h]
    \phantom{}\hfill%
    \IncludeGraphics[page=1]{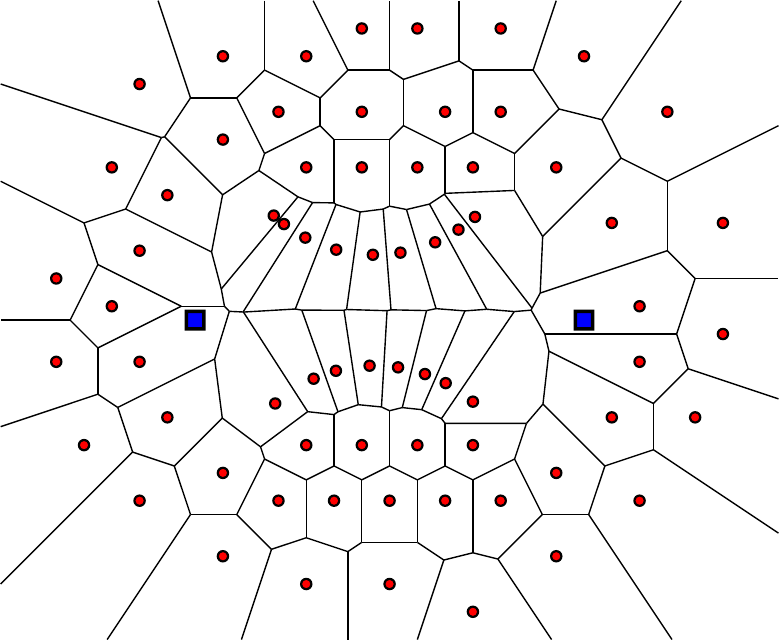}%
    \hfill%
    \IncludeGraphics[page=2]{figs/middle_is_better}%
    \hfill\phantom{}
    \caption{Inserting two middle sites is enough. Here we insert both
       the two endpoints into the diagram, and two middle sites (i.e.,
       the two red cells). Note that one can secure an existing site,
       instead of inserting a new one -- this has the same cost. In
       this specific example, if one uses only existing cells, the
       price is six (the green cells).}%
    \figlab{two}
\end{figure}

\paragraph{Formal problem statement.}

Let $\PS$ be a set of $n$ points in the plane, and let $\VorX{\PS}$
denote its Voronoi diagram.  Two points $\ps,\pt$ in the plane can
directly communicate \emphw{safely}, if $\ps$ and $\pt$ are ``close''
to each other. Formally, we require that the cells of $s$ and $t$ are
adjacent in the Voronoi diagram of $\PS \cup \{\ps,\pt\}$ --
geometrically, this corresponds to the existence of a disk that
contains $\ps$ and $\pt$, and no other points of $\PS$.  Naturally,
most points can not communicate safely. To overcome this, one can
insert a set of points $\QS = \{\pq_1, \ldots, \pq_\iend\}$ into the
$\PS$, such that $\pq_0 = \ps$, $\pq_\iend = \pt$, and in the new
diagram $\VorX{\PS \cup \QS}$ the point $\pq_i$ can safely communicate
with $\pq_{i+1}$, for all $i$. It is natural to ask for the minimum
size set $\QS$, such that $\ps$ and $\pt$ can communicate safely, via
$\iend-1$ hops. Note, that a site $\pq_i$ might be an existing site --
this can be interpreted as securing this site (this costs the same as
inserting a new site).

A naive approach is to insert $\ps$ and $\pt$ into the initial Voronoi
diagram, and compute the shortest path between them in the resulting
dual graph (i.e., the Delaunay triangulation of
$\PS \cup \{ \ps, \pt \}$) -- this corresponds to adding all the
intermediate nodes in this path to $\QS$.  However, there are natural
scenarios, see \figref{two}, where allowing the sites to be placed
arbitrarily in the plane significantly reduces the number of sites
needed.  As mentioned earlier, the price of securing an existing site,
or introducing a new site is the same.

\paragraph{Our results.}
We describe how the region of reachable points after $t$ insertions
looks like, and show how to compute it efficiently. We then describe
an $O(n \log n)$ time algorithm for computing the shortest such path
between two points, using a process that is a variant of the Dijkstra
algorithm. We implemented the algorithm using \CGAL, and the
source code is available online \cite{hv-scias-21}.

\paragraph{Related results.}

The basic algorithm is a variant of continuous Dijkstra, a technique
that was used for shortest path algorithms among obstacles, and for
shortest path on polyhedral surfaces in three dimensions
\cite{mmp-dgp-87}. The basic process of inserting points is similar in
nature to Delaunay refinement \cite{r-draqt-95}. The new algorithm can
be viewed as a combination of these two techniques.

\section{The algorithm}

\paragraph{Notations.}

For a point $\pp$ and a radius $r$, let $\diskY{\pp}{r}$ denote the
disk of radius $r$ centered at $\pp$.  For a set of objects $\XX$, in
the plane, let $\VorX{\XX}$ denote the Voronoi diagram of $\XX$. For
an object $\obj$, let $\CellY{\obj}{\XX}$ denote the cell of $\obj$ in
the Voronoi diagram of $\{ \obj \} \cup \XX$.

\subsection{Understanding the problem: %
   The flowering process}

As an example, consider a point set formed by a hexagonal lattice --
see \figref{hex}.  Let $\pq_0 = \ps$ be a starting site (which is one
of the points of the lattice), and assume our purpose is to reach a
site $\pq_\iend = \pt \in \PS$.  The site $\ps$ can communicate
directly with all the sites that are adjacent to its Voronoi cell
$\RR_0 = \CellY{\ps}{\PS}$, without requiring the insertion of any
middle sites.

\begin{figure}[t]
    \centering%
    \IncludeGraphics[page=1]{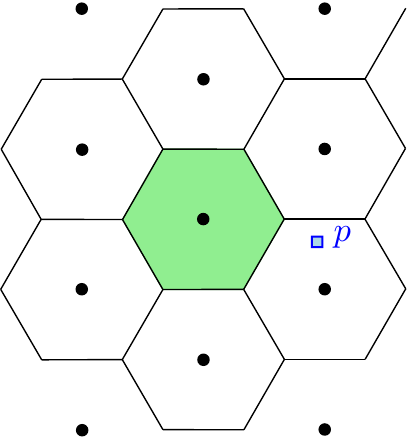}%
    \quad%
    \quad%
    \IncludeGraphics[page=3]{figs/hex_2}%
    \quad%
    \quad%
    \IncludeGraphics[page=4]{figs/hex_2}
    \caption{The hexagonal case.}
    \figlab{hex}
\end{figure}

\begin{figure}[h]
    \phantom{}%
    \hfill%
    \IncludeGraphics[page=3,width=0.32\linewidth]{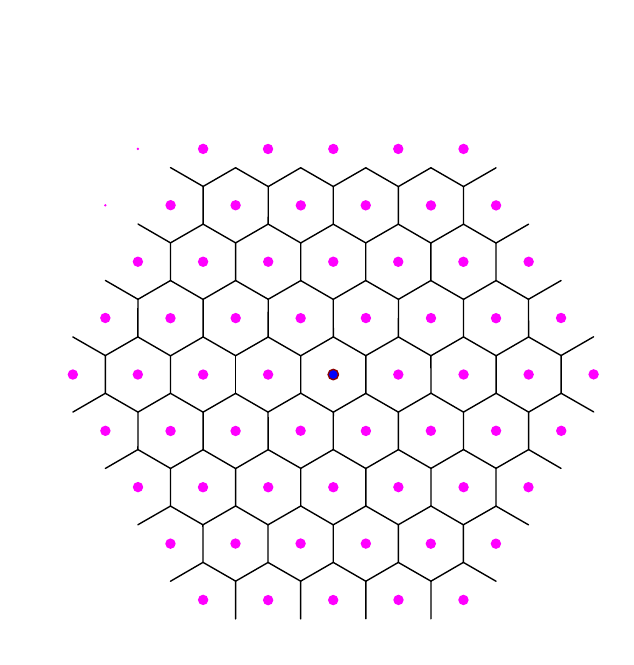}%
    \hfill%
    \IncludeGraphics[page=4,width=0.32\linewidth]{figs/flower}%
    \hfill%
    \IncludeGraphics[page=5,width=0.32\linewidth]{figs/flower}%
    \hfill%
    \phantom{}%

    \vspace{-1.5cm}%

    \phantom{}%
    \hfill%
    \phantom{}%
    \hfill%
    \IncludeGraphics[page=6,width=0.32\linewidth]{figs/flower}%
    \hfill%
    \IncludeGraphics[page=7,width=0.32\linewidth]{figs/flower}%
    \hfill%
    \phantom{}%
    \hfill%
    \phantom{}%

    \vspace{-0.5cm}%

    \caption{The flowering process of how the shortest path region
       grows.}
    \figlab{hex:2}
\end{figure}

If we insert a new middle site $\pq_1$, then its cell must touch the
cell of $\RR_0$. Consider such a location $\pp$ for a middle site. For
the new site of $\pp$ to be adjacent to $\RR_0$, after $\pp$ is
inserted, there must be a point $\pv \in \RR_0$, such that $\pv$ is as
close (or closer) to $\pp$ than it is to $\pq_0 = \ps$. This region of
influence, where $\pv$ can be ``occupied'', is exactly the disk
$\diskY{\pv}{\dY{\pv}{\ps}}$.  See \figref{hex}. As such, the
allowable region to insert this first site, is
\begin{equation*}
    \RU_1 = \cup_{\pv \in \RR_0} \diskY{\pv}{\dY{\pv}{\ps}\bigr.}.
\end{equation*}
(Namely, sites inserted outside $\RU_1$ are not going to be adjacent
to $\RR_0$ in the new Voronoi diagram.)  The set $\RU_1$ is a union of
disks -- in this specific case, it is equal to the union of disks
placed at vertices of the original Voronoi cell, see \figref{hex}.
Indeed, disks centered in the interior of $\RR_0$ are covered by disks
centered at the boundary of the Voronoi cell. Importantly, the union
of disks centered at a Voronoi edge forms a pencil which is covered by
the two disks induced by the two Voronoi vertices of the edge, see
\figref{pencil}.

\noindent%
\begin{minipage}{0.99\linewidth}
    \bigskip%
    \begin{minipage}{0.5\linewidth}
        \hfill{\IncludeGraphics{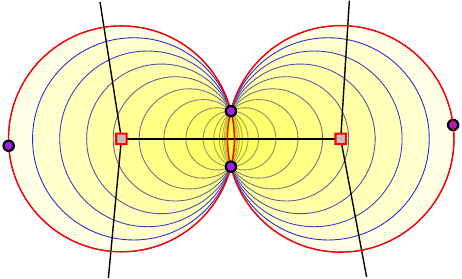}}\hfill\phantom{}
    \end{minipage}
    \begin{minipage}{0.5\linewidth}
        \bigskip%
        \captionof{figure}{A pencil of disks. The union of the two
           extreme disks of the pencil (i.e., the one induced by the
           Voronoi vertices), cover the union of the disks in the
           pencil.}%
        \figlab{pencil}
    \end{minipage}
    \bigskip%
\end{minipage}

Let $\PS_1 = \PS \setminus \RU_1$.  The region that can potentially be
covered by such a single insertion of a site into $\RU_1$, belongs to
the region
\begin{equation*}
    \RR_1 = \bigcup_{\pp \in \RU_1} \CellY{\pp}{ \PS_1},
\end{equation*}
The set $\RR_1$ contains all the points in the plane that are closer
to $\RU_1$ than to any point in $\PS_1$. Namely,
$\RR_1 = \CellY{ \RU_1}{ \PS_1}$. Any site $\pu \in \PS_1$ that its
Voronoi cell is adjacent to $\RR_1$ in
$\VorX{ \{\RU_1\} \cup \PS_1 }$, can communicate with $\pq_0$ via the
insertion of a single site.

As such, we may treat $\RU_1$ simply as the union of disks at the
finitely many vertices of $\RR_0$. The boundary of $\RU_1$ is all that
is necessary to accurately compute the boundary of $\RR_1$ (which may
be viewed as a wavefront computed by the algorithm). The set $\RR_1$
is then the union of Voronoi cells in a Voronoi diagram of disks and
points. This underlying diagram is therefore an additive weighted
Voronoi diagram, and its boundary edges are portions of hyperbolas and
straight segments.  See \figref{hex:2}.

Formally, let the set $\RR_i$ be the set reachable by inserting $i$
middle sites starting at $\ps$. The set $\RU_{i+1}$ is the region
where one can insert an $(i+1)$\th point, $\pq_{i+1}$, such that one
can form a connected, safe chain back to $s$. By the discussion above,
$\RU_{i+1}$ can be treated as the union of finitely many disks at the
vertices of $\RR_i$, with some additional care.  Let
$\PS_{i+1} = \PS \setminus \RU_i$.  The $(i+1)$\th{} \emphw{occupied}
region is
\begin{equation}
    \RU_{i+1} = \cup_{\pv \in \RR_i} \diskY{\pv}{\dSetY{\pv}{ \PS
          \setminus \RR_i }\bigr.},%
    \eqlab{insert}%
\end{equation}
where $\dSetY{\pv}{\PS} = \min_{\pp \in \PS} \dY{\pv}{\pp}$.  And the
$(i+1)$\th \emphi{reachable} region is
\begin{equation*}
    \RR_{i+1} = \bigcup_{\pp \in \RU_{i+1}} \CellY{\pp}{ \PS},
\end{equation*}
As soon as the region $\RR_{i+1}$ shares a boundary with the Voronoi
cell of $\pt$, one can stop, and reconstruct the safe path.

\subsection{The basic algorithm}
The input is a set $\PS$ of $n$ points in the plane, and two point
$\ps, \pt \in \PS$. The algorithm initially computes the Voronoi
diagram of $\PS$, and has a queue $\Queue$ of sites, which is
initially set to $\ps$. A site here is a disk (the initial point is a
disk of radius $0$, naturally).

The algorithm works in rounds. In each round, it extracts all the
sites in the queue, and inserts them into the current additive
weighted Voronoi diagram -- these are the outer disks, whose outer
boundary forms the boundary of $\RU_i$ in the $i$\th round.  Next, the
algorithm scans all the inserted sites, and looks at the adjacent
Voronoi vertices of their cells. All of the Voronoi vertices that are
outside $\RU_i$ (how to check for this condition is described below),
are inserted into the queue to be handled in the next round.

Once the wavefront reaches $\pt$, the process stops -- or
alternatively, we compute this diagram till the whole plane is
covered, and preprocess the resulting map for point location.

\medskip%
\noindent%
\begin{minipage}{\linewidth}%
    \hfill%
    \IncludeGraphics[page=1]{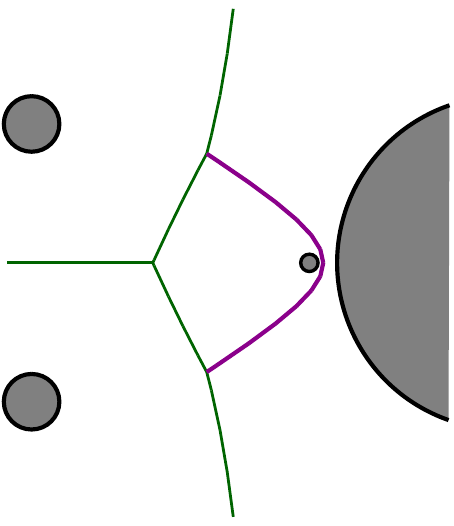} \hfill%
    \IncludeGraphics[page=2]{figs/pocket}%
    \hfill%
    \phantom{} \captionof{figure}{The Voronoi diagram of points and
       disks.  A hyperbolic segment in this diagram, and its pencil of
       disks. Only the two disks, associated with the Voronoi
       vertices, contribute to the outer boundary of the union of the
       disks.}  \figlab{h:segment}
\end{minipage}

\noindent%
\begin{minipage}{\linewidth}
    \centerline{\IncludeGraphics{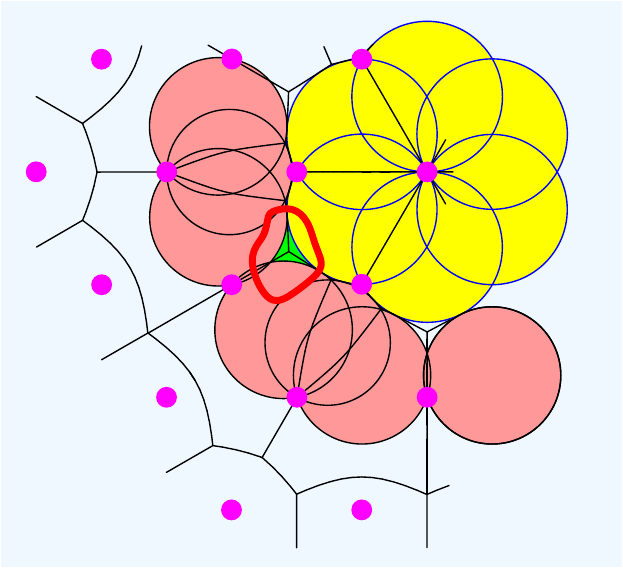}}%
    \captionof{figure}{Two generations of disks inserted, and a pocket
       that is left behind. This pocket corresponds to a middle of a
       hyperbolic edge from previous generation, which was covered by
       the inserted disks.  } \figlab{pocket}
\end{minipage}

\paragraph{What disks to insert.}

The algorithm inserts the disks of \Eqref{insert}.  As discussed
earlier, one needs to insert only the disks that have their center on
the boundary of $\RR_i$. This boundary is a union of hyperbolic
segments.  Each such hyperbolic segment corresponds to a pencil of
disks that are included in the union, but the algorithm inserts only
the two extreme disks that corresponds to the two endpoints of the
hyperbolic segment (they are vertices of the Voronoi diagram).  See
\figref{h:segment}.

Each such hyperbolic segment of the boundary might give rise to a
pocket that is left behind the front. Since these pockets cannot
contribute to the front -- we leave them behind unrefined. Such a
pocket is depicted in \figref{pocket}.

To correctly leave the pockets behind, the algorithm remembers for
each site inserted the layer (i.e., generation) of the propagation it
came from. As such, when inspecting a Voronoi vertex, the sites that
gave rise to it and their generation are known. All sites that are
from two or more generations ago are not inserted, thus blocking the
wavefront from propagating backwards.

\paragraph{Reconstructing the shortest path.}

The above insertion process creates ``rings'' of inserted disks, where
each generation forms a single ring, see \figref{rand} for an
example. Importantly, every inserted disk of a certain generation
touches a disk of a previous generation, where a disk of the first
generation touches the source vertex. As such, for every such disk
there is a chain of touching disks that goes back to the source. This
is the \emphi{insertion path} of this disk.

As such, when the wavefront arrives at the target cell, then in the
resulting Voronoi diagram there is a hyperbolic bisection segment of
the Voronoi diagram, that separates the target cell from some newly
inserted disks. Place any disk $D$ centered at this hyperbolic curve,
that touches the target point $t$.  The disk $D$ touches a disk $D'$
that was just inserted. The disk $D$ together with the insertion path
forms a sequence of disks that touch each other, and touch $s$ and
$t$. The key observation is that if we insert the points of tangency
between two consecutive disks in this sequence into the original set
of points, in the resulting diagram, the newly inserted Voronoi cells
form a connected safe zone, with the fewest number of insertions, as
desired.  We note that the choice of points to be inserted is not
unique.  The process is illustrated in \figref{trace:back}.

\medskip%
\noindent%
\begin{minipage}{\linewidth}
    \IncludeGraphics[page=1]{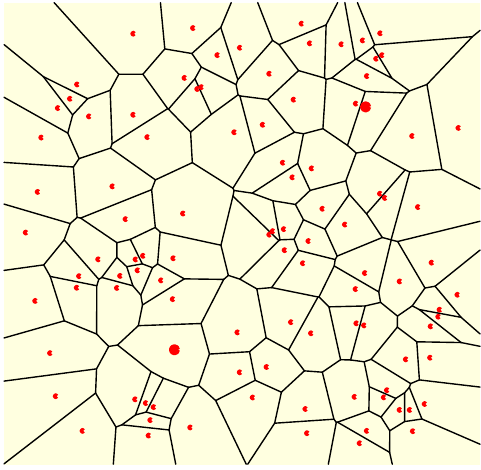} \hfill%
    \IncludeGraphics[page=2]{figs/trace_back} \hfill%
    \IncludeGraphics[page=3]{figs/trace_back} \captionof{figure}{}
    \figlab{trace:back}
\end{minipage}

\subsection{Running time analysis}

A point $\pp$ of the input is \emphi{active} at time $i$, if it is
adjacent to a site inserted in the $i$\th iteration.

\begin{lemma}
    \lemlab{before:after}%
    An input point can be active for at most two iterations.
\end{lemma}
\begin{proof}
    An input point $\pp$ is \emphi{discovered} at time $i$, if its
    Voronoi cell $\Cell$ in
    $\VorC_i = \VorX{ \{\RU_i\} \cup (\PS \setminus \RU_i) }$ is
    adjacent to the cell of $\RU_i$ in this diagram.

    Informally, once a cell is being discovered, in the next iteration
    of insertions, disks would inserted that touch it, and the
    following iteration it would be blocked from the (wave) front.
    Namely, the boundaries of $\Cell$ and $\RR_i$ intersect, as
    $\RR_i$ is the cell of $\RU_i$ in $\VorC_i$. As such, $\pp$ is on
    the boundary of $\RU_{i+1}$, or $\pp$ might be contained in the
    interior of $\RU_{i+1}$. Indeed, consider a point $\pq$ on the
    common boundary between $\Cell$ and $\RR_i$, and consider the disk
    of radius $\dY{\pp}{\pq}$ centered at $\pq$, and observe that this
    disk has $\pp$ on its boundary, and it is contained in
    $\RU_{i+1}$, see \Eqref{insert}. But that implies that $\Cell$ is
    contained in the Voronoi cell of $\RU_{i+1}$ in
    $\VorC_{i+1} = \VorX{ \{\RU_{i+1}\} \cup (\PS \setminus \RU_{i+1})
    }$. This cell is thus $\RR_{i+1}$. Thus, the cell of $\pp$ might
    interact with cells created in $\RU_i$, and $\RU_{i+1}$. Clearly,
    $\pp$ cannot be adjacent to cells inserted in later iterations.
\end{proof}

Let $n_i$ be the number of the active input points at time $i$.  Since
inserting a disk, updating the Voronoi diagram, and discovering the
new Voronoi sites are all done locally, and can be charged to new
entities created or deleted, the following is straightforward to
verify. It is critical here that we are inserting sites centered at
Voronoi vertices, whose location we already know -- that is, there is
no need to perform a point-location query for the insertion. As such,
we get the following.

\begin{lemma}
    \lemlab{wave}%
    The total running time of the $i$\th iteration, is bounded by
    $O(\sum_{j=i-2}^{i+2} n_j)$.
\end{lemma}

\begin{proof}
    We only sketch the proof. The quantity stated above bounds the
    number of sites that the wavefront might interact with in the
    $i$\th iteration. In particular, it bounds the number of sites
    inserted at time $i$.  By \lemref{before:after}, a site is active
    only for a constant number of iterations. A new site inserted
    which is adjacent only to inserted sites, is a pocket, and the
    algorithm does not insert it. As such, inserted sites must be
    adjacent to original input points. It is easy to verify that an
    input site can support only a constant number of such sites around
    it, which readily implies that while a site is active, only a
    constant number of new sites inserted might be charged to it.

    This implies by planarity, that the total complexity of the
    Delaunay triangulation in the $i$\th iteration is proportional the
    total complexity of the input points in the adjacent layers, which
    implies the claim.
\end{proof}

\begin{theorem}
    The running time of the algorithm is $O(n \log n)$.
\end{theorem}
\begin{proof}
    Reading the input and computing the Voronoi diagram takes
    $O(n \log n)$ time. By \lemref{wave}, the total running time of
    the later stages is proportional to $\sum_i O(n_i) = O(n)$.
\end{proof}

\section{Implementation and some pictures}

The algorithm was implemented in \texttt{C++} using CGAL
\cite{cgal-20}. Specifically, we use the 2D Apollonius Graphs
implementation (which was implemented by Menelaos Karavelas and
Mariette Yvinec), see \cite{k-2dvda-20,ky-2ag-20}.

The source is available at bit{}bucket \cite{hv-scias-21}. The
repository also includes some input files, and a script to run the
program on various provided inputs.

\figref{rand} illustrates the execution of the algorithm on a random
input, \figref{illinois} is for a real world input, which is a point
set of locations in Illinois (downloaded from the US government
census).  \figref{hex:exec} illustrates some results from executing
the algorithm on a hexagonal grid.  \figref{results} summarizes the
inputs tested.  The point sets that are generated are of three types:
diagrams of regular hexagonal Voronoi cells, diagrams of randomly
chosen points, and diagrams of points representing the state of
Illinois constructed from census data.

Bounding points along the edges of a bounding box of the point set are
added for the underlying additive diagram phase of the algorithm, to
ensure no edges are infinite. The start and end points are chosen by
considering a rectangular subset of points in the interior and
choosing the furthest two points among them. Then the algorithm is run
until a path is found safely connecting the two cells.

Potential floating point or degeneracy issues are avoided by adding
small perturbations to the point set as points are inserted into the
initial diagram.  Radii of the disks being inserted in the additive
diagram phase are also dilated by a very small factor for the same
reasons.

While we are not reporting the running times, they seem to be near
linear, and agree with the theoretical analysis (ignoring the initial
construction time).

We used the CGAL Apollonius Graph hierarchy for the
implementation. The provided library insertion function has two phases
-- locating a nearest neighbor, and then constructing a new Voronoi
cell by identifying conflict edges, etc. The nearest neighbor location
in our algorithm can theoretically be done in constant time, because
all the points we insert are Voronoi vertices whose adjacent sites are
immediately known. However, we did not bypass the internal nearest
neighbor search because the results of using the Apollonius Graph
hierarchy were already linear time in practice.

The running times seems to be the same on different point sets, of the
same size. However, on regular hex meshes the BFS and algorithm paths
have the same length, and small improvements in practice are noticed
because of the perturbations we apply and when there are a large
number of points. For other point sets, the resulting paths are
significantly shorter than the paths provided by only using existing
sites.

\paragraph*{Acknowledgment.}
The authors would like to thank Alon Efrat -- long time ago he
mentioned this problem to the first author.

\bibliographystyle{salpha}%
\bibliography{voronoi_ins}

\begin{figure}[p]
    \begin{tabular}{cc}
      \IncludeGraphics[width=0.45\linewidth]{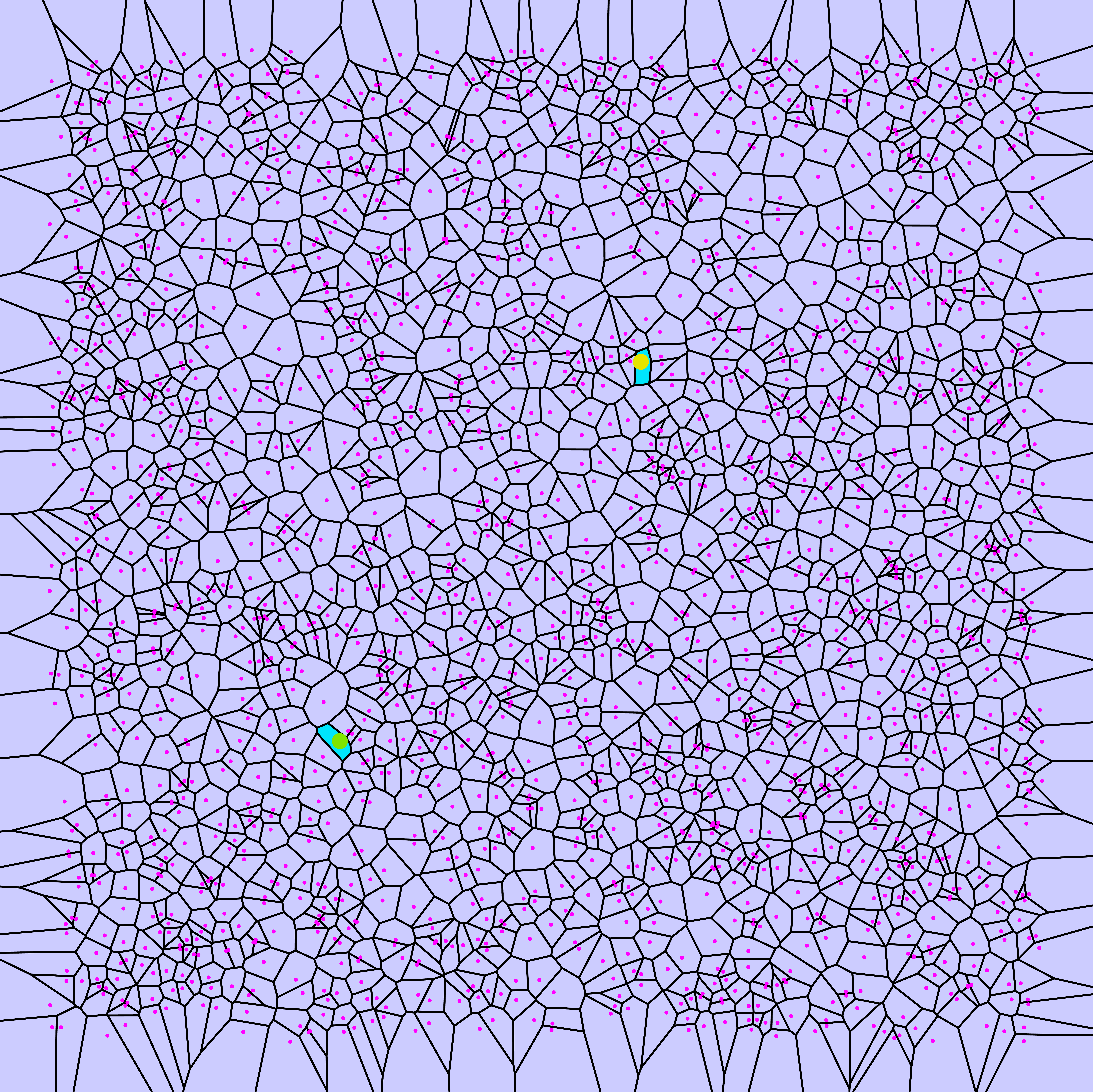}
      &
        \IncludeGraphics[width=0.45\linewidth]{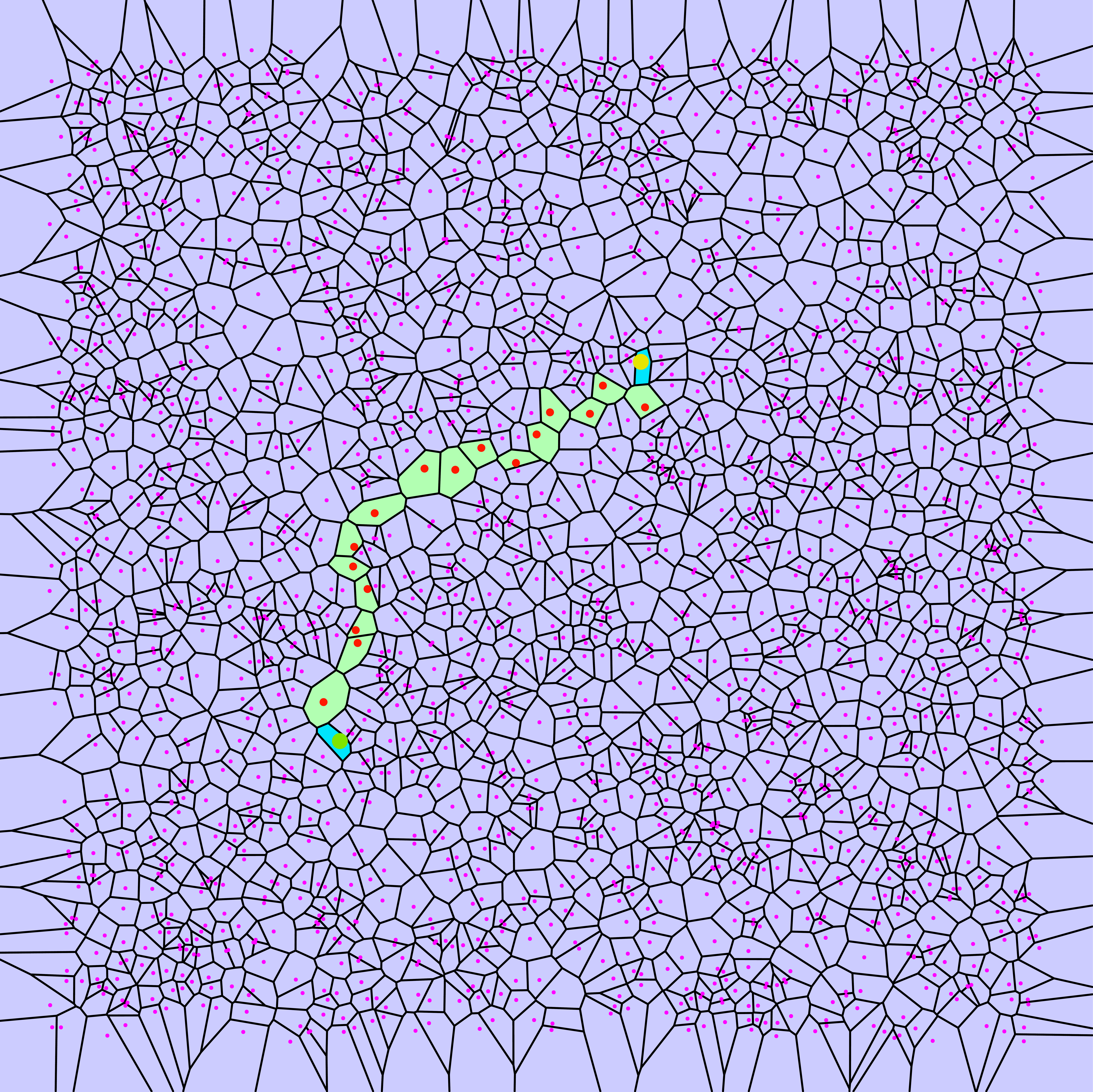}
      \\
      Input & BFS shortest path (16 intermediate points)\\[0.2cm]
      \IncludeGraphics[width=0.45\linewidth]{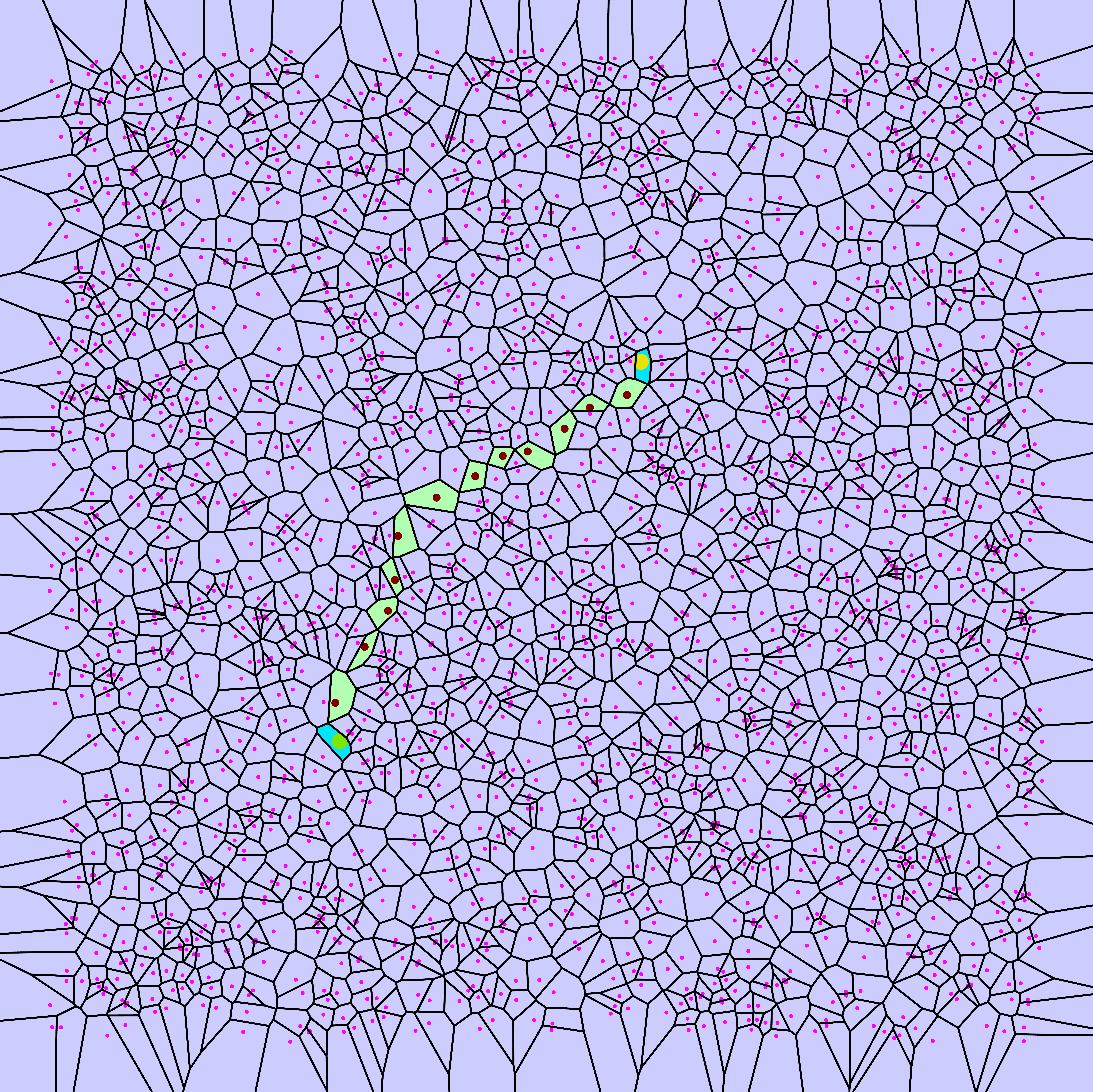}
      &
        \IncludeGraphics[width=0.45\linewidth]{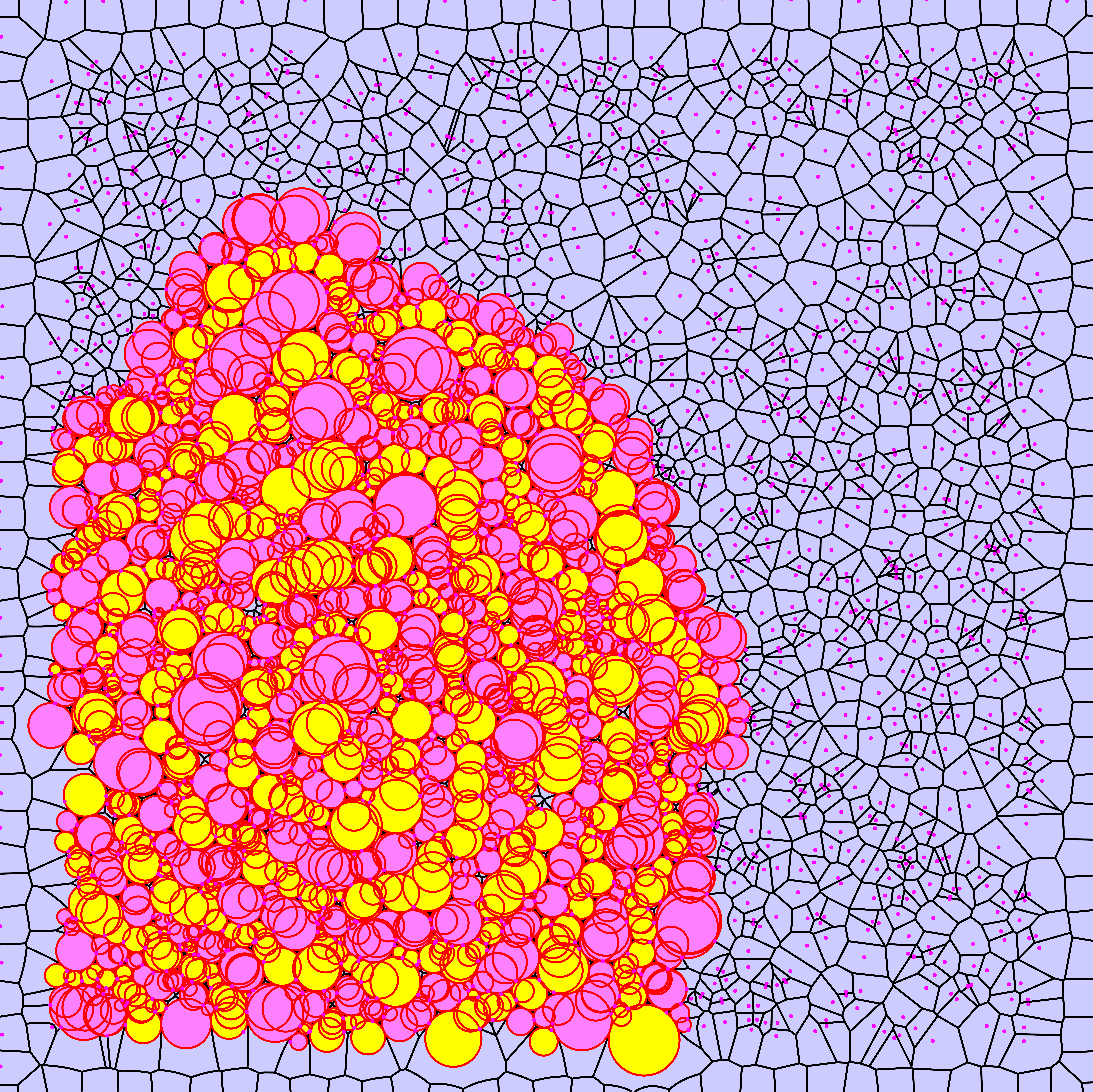}
      \\
      Algorithm shortest path (12 intermediate points) & Sites inserted during  execution.\\
    \end{tabular}
    \caption{A random point set with 2000 points, and the algorithm
       execution on it.}
    \figlab{rand}
\end{figure}

\begin{figure}[p]
    \begin{tabular}{cc}
      \IncludeGraphics[width=0.45\linewidth]{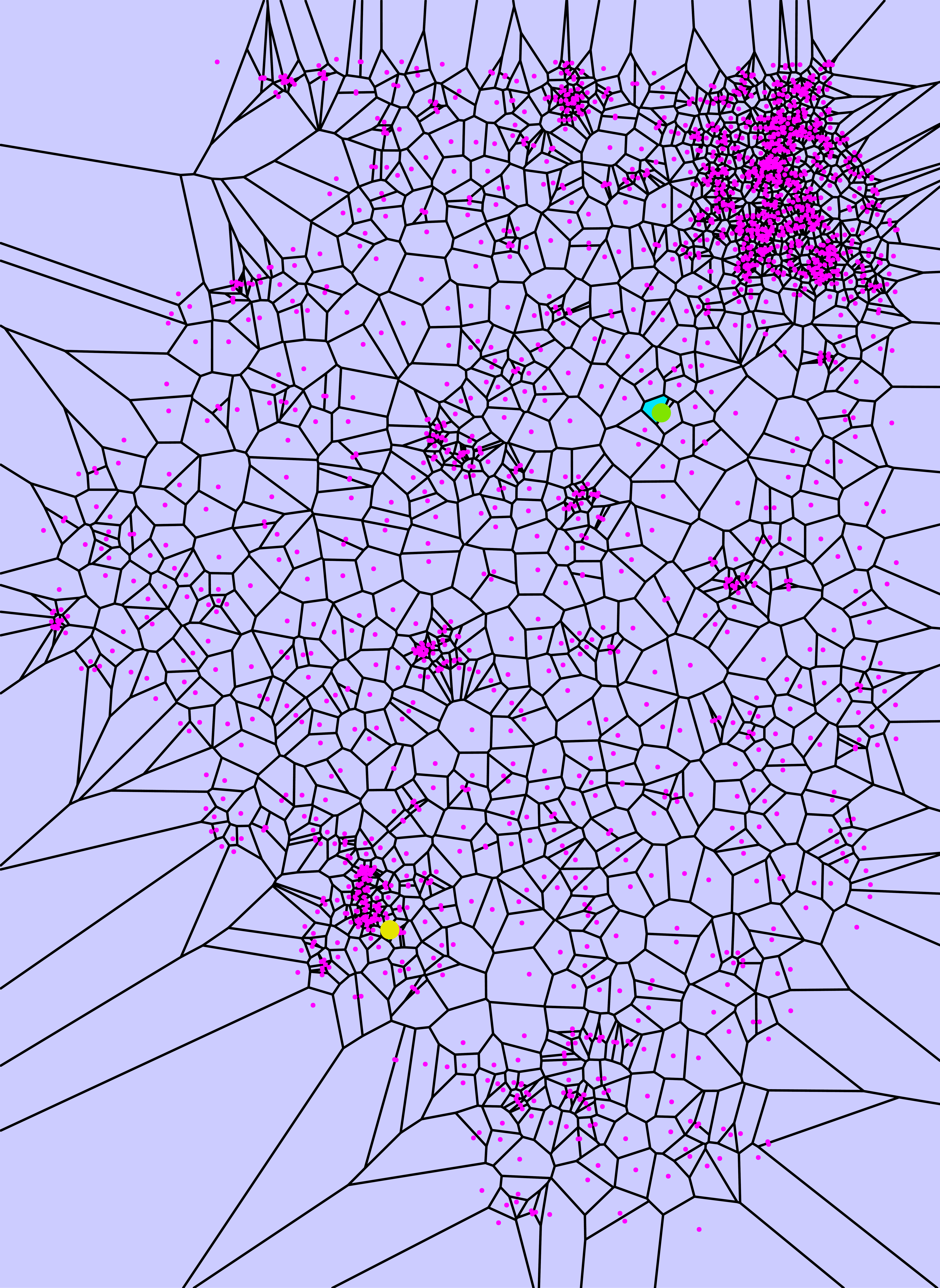}
      &
        \IncludeGraphics[width=0.45\linewidth]{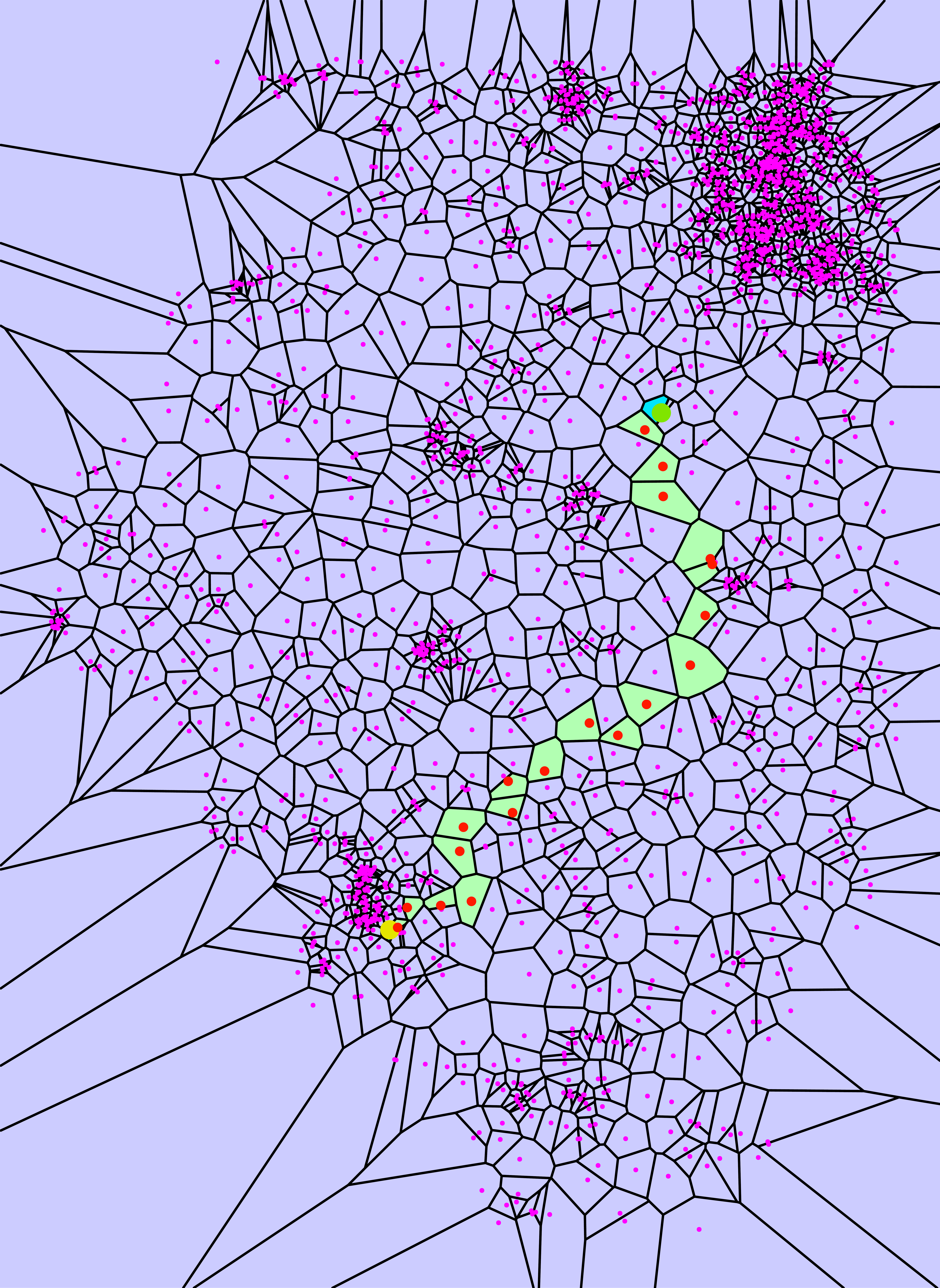}
      \\
      Input & BFS shortest path (19 intermediate points)\\[0.2cm]
      \IncludeGraphics[width=0.45\linewidth]{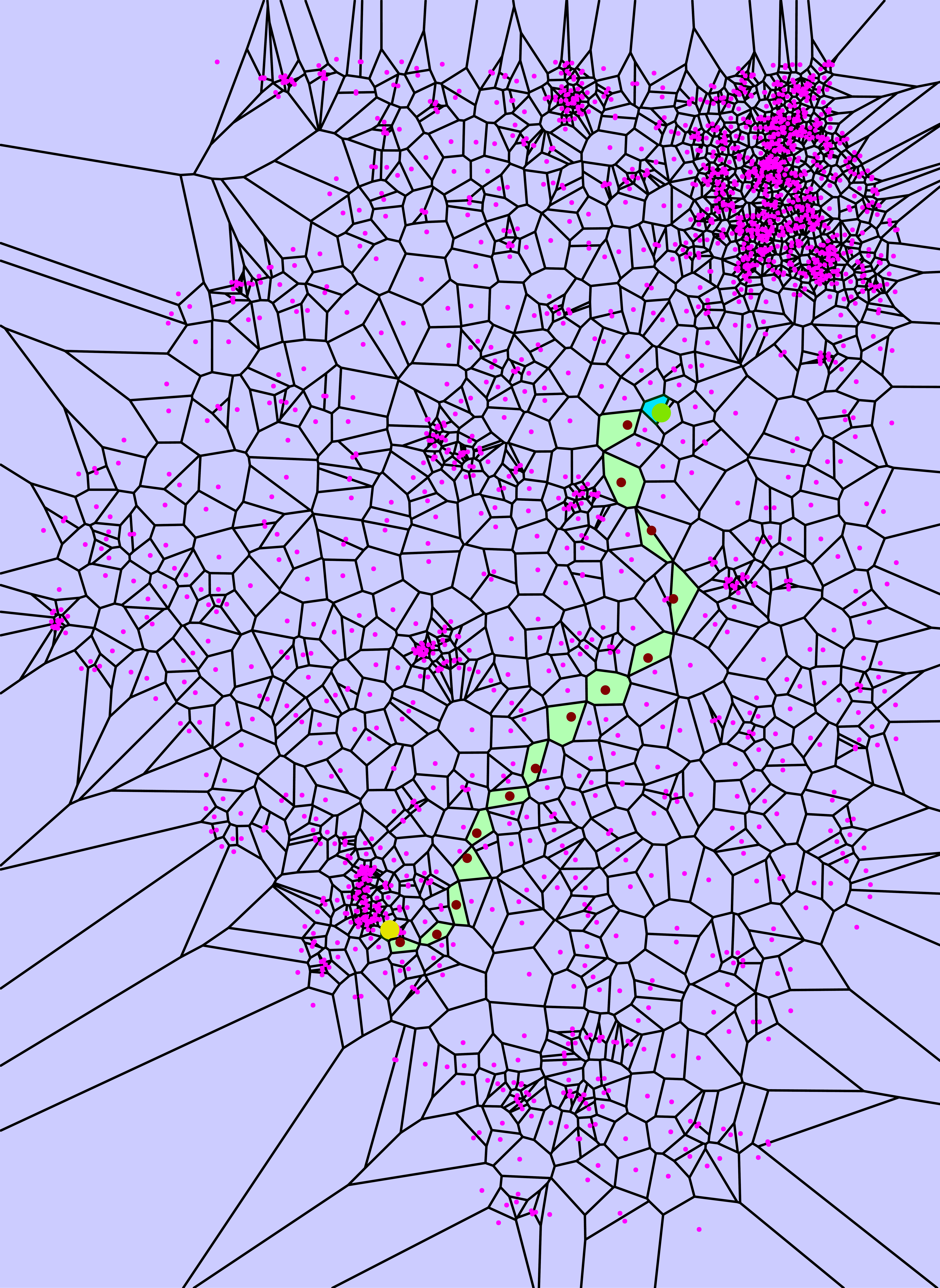}
      &
        \IncludeGraphics[width=0.45\linewidth]{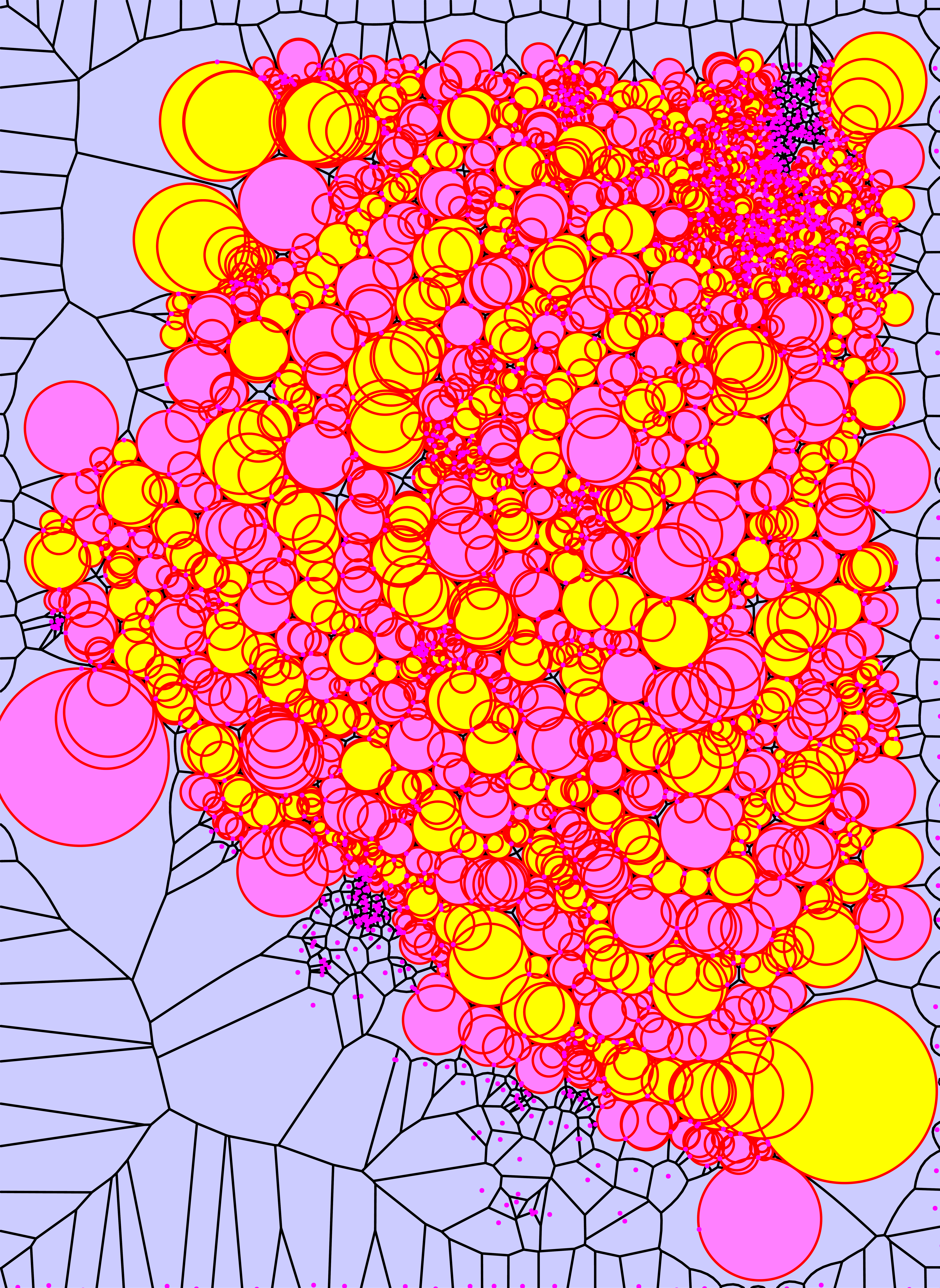}
      \\
      Algorithm shortest path (14 intermediate points) & Sites inserted during  execution.
    \end{tabular}
    \caption{A point set made out of locations in Illinois (downloaded
       from the census and sparsified), and the algorithm execution on
       it.}
    \figlab{illinois}
\end{figure}

\begin{figure}[p]
    \begin{tabular}{cc}
      \IncludeGraphics[width=0.45\linewidth]{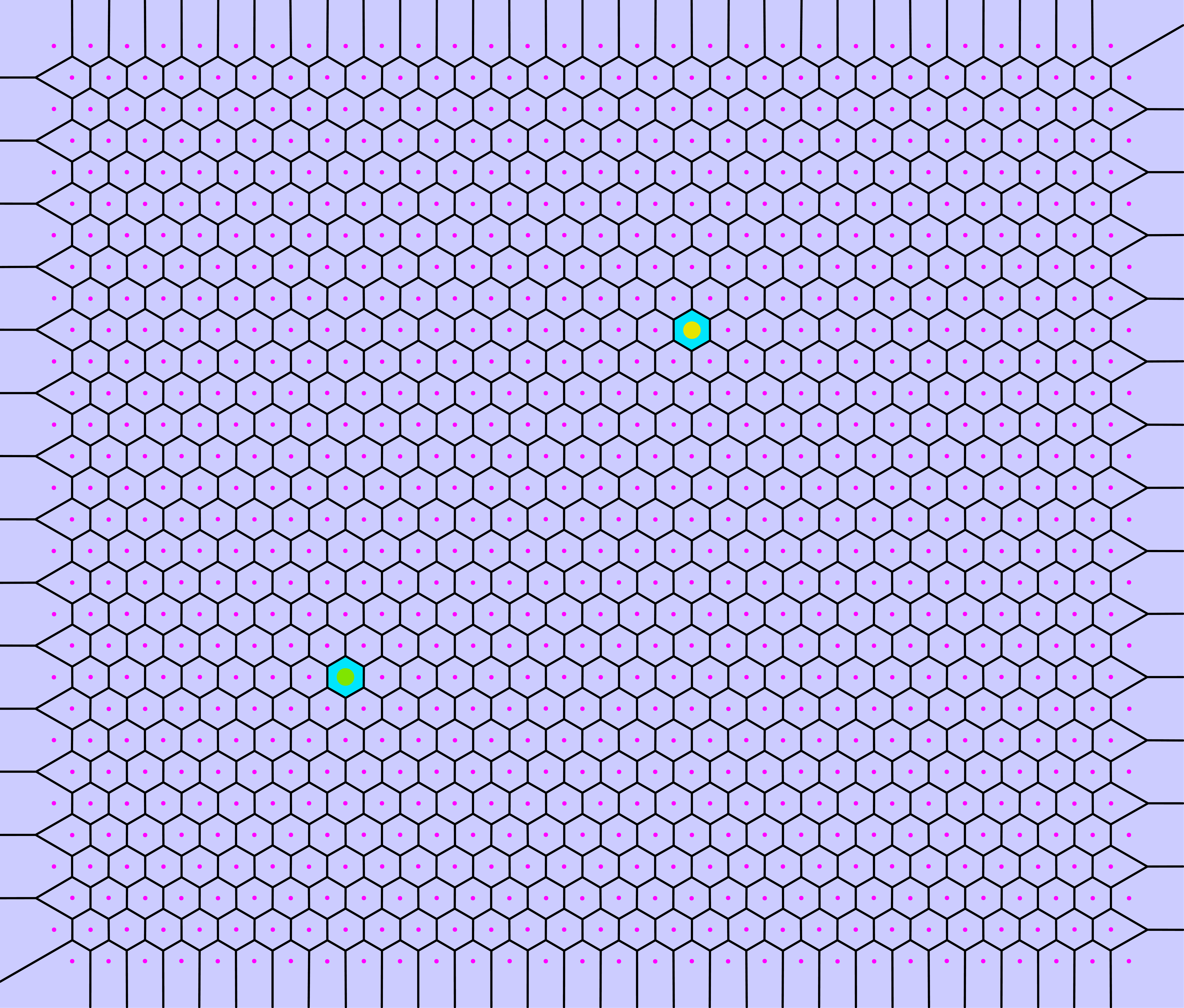}
      &
        \IncludeGraphics[width=0.45\linewidth]{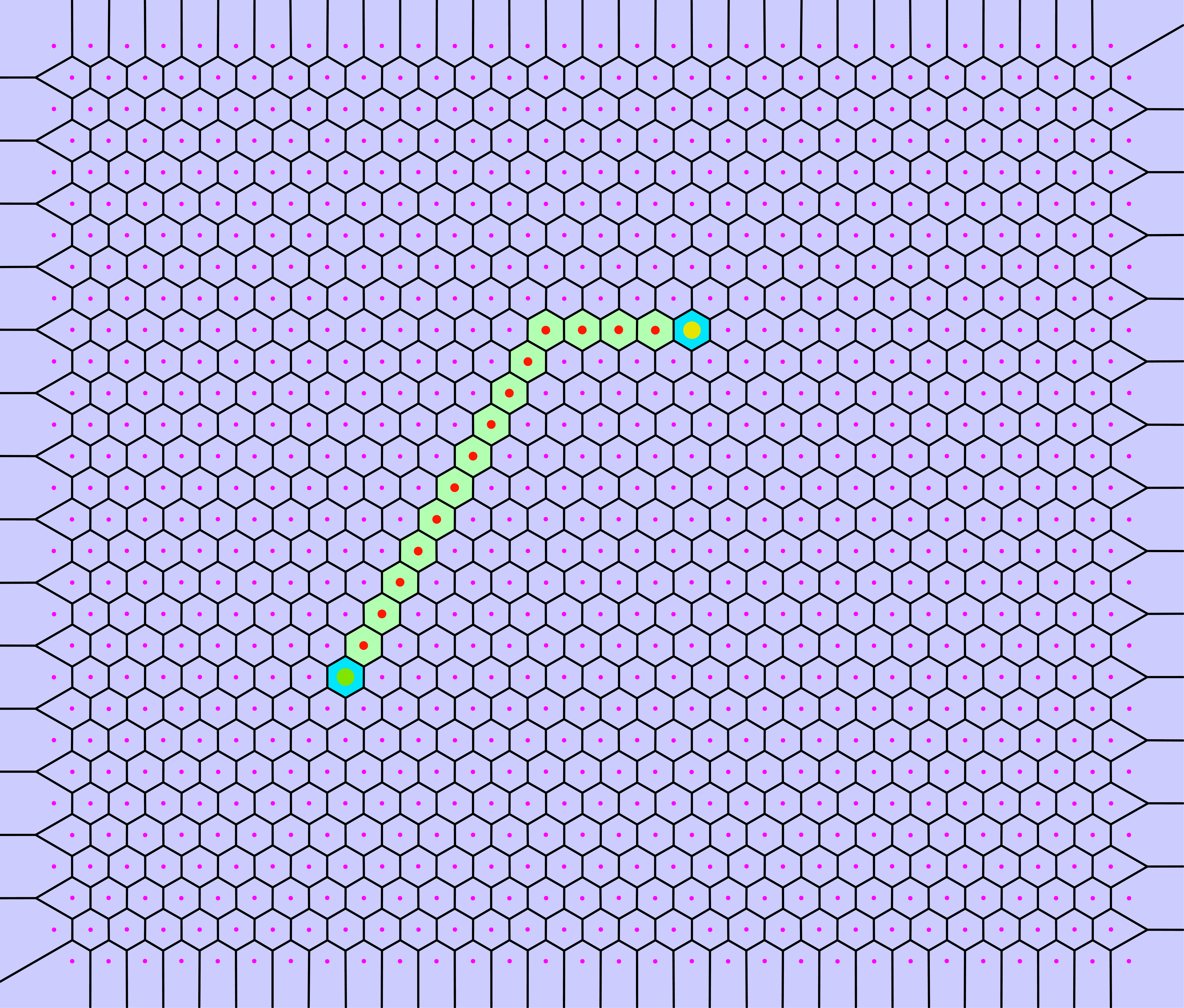}
      \\
      Input & BFS shortest path (14 intermediate points) %
      \\[0.2cm]
      \IncludeGraphics[width=0.45\linewidth]{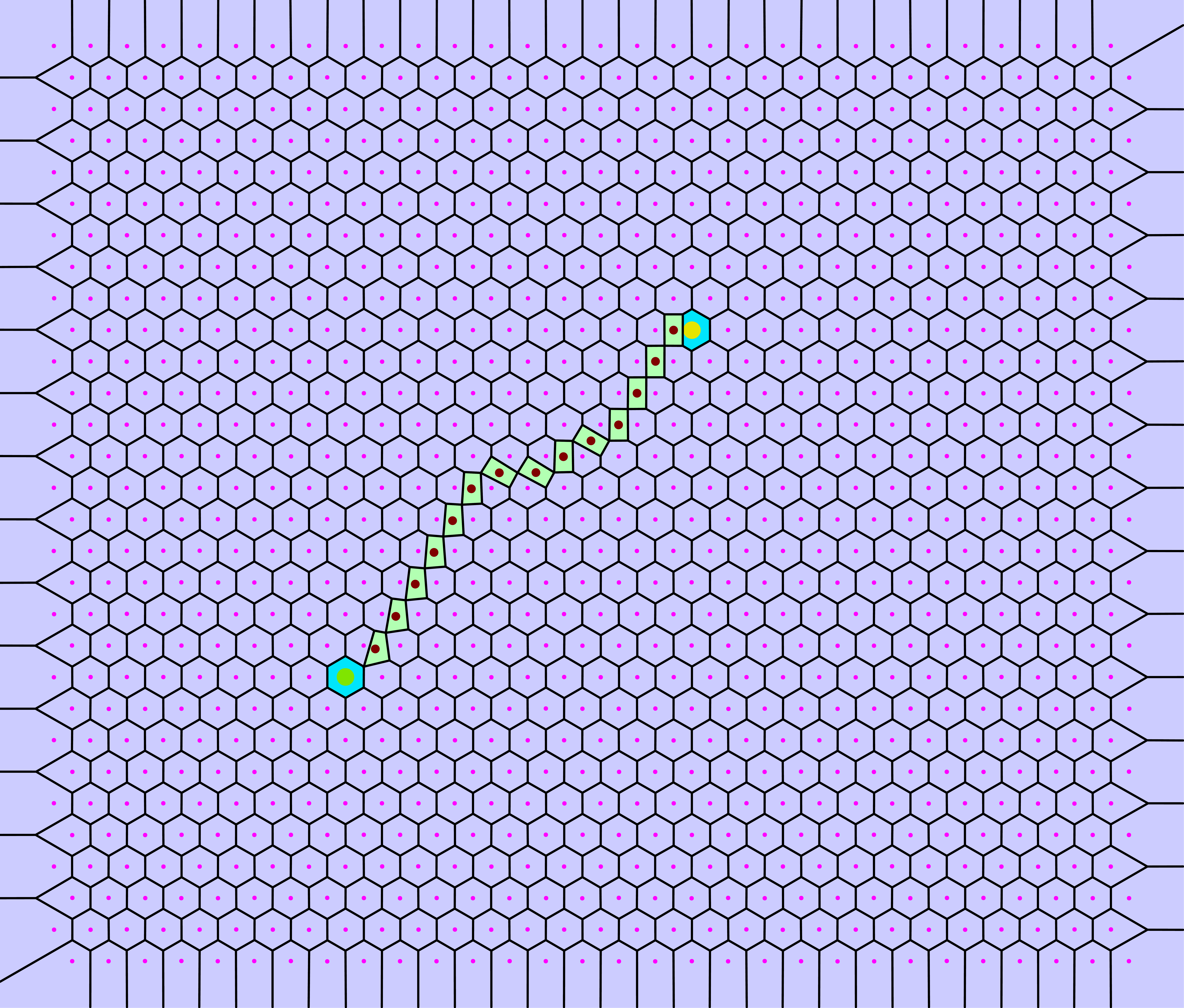}
      &
        \IncludeGraphics[width=0.45\linewidth]{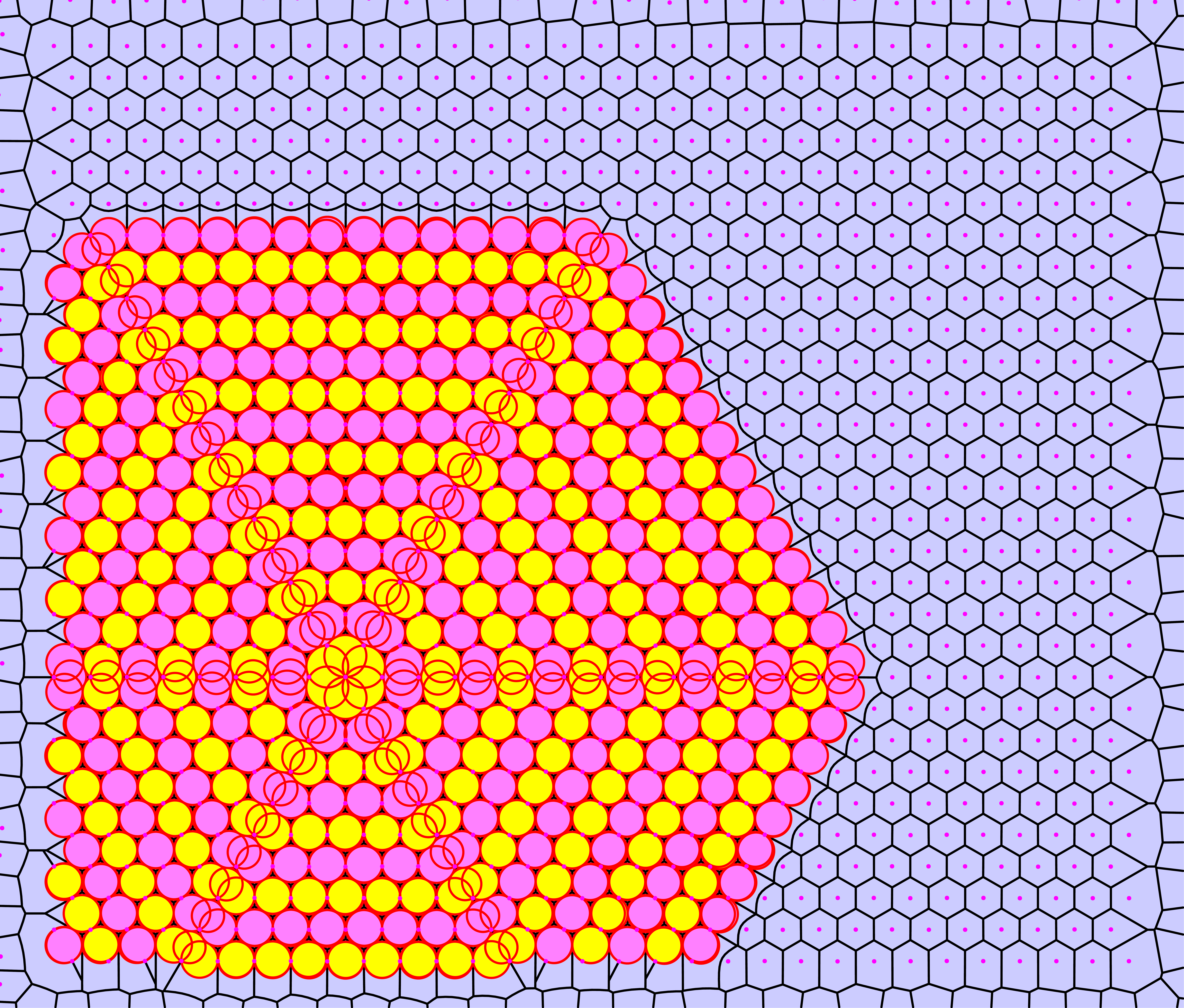}
      \\%
      Algorithm shortest path (14 intermediate points) & Sites inserted during  execution.
    \end{tabular}
    \caption{A point set made out of a hexagonal grid, and the
       algorithm execution on it.}
    \figlab{hex:exec}
\end{figure}

\begin{figure}
    \centerline{%
\begin{tabular}{|l|r|r|r|}
\hline
\hline
Input name & \# points & BFS path & Algorithm path len \\
\hline
   \hline
    hex\_010 & 100 & 5 & 5 \\
   \hline
    hex\_030 & 900 & 18 & 18 \\
   \hline
    hex\_060 & 3,600 & 56 & 54 \\
   \hline
    hex\_090 & 8,100 & 63 & 61 \\
   \hline
    hex\_120 & 14,400 & 77 & 75 \\
   \hline
   \hline
    i\_d\_02 & 17,560 & 55 & 38 \\
   \hline
    i\_d\_04 & 8,816 & 39 & 28 \\
   \hline
    i\_d\_08 & 4,435 & 30 & 23 \\
   \hline
    i\_d\_16 & 2,242 & 22 & 15 \\
   \hline
    i\_d\_32 & 1,123 & 17 & 12 \\
   \hline
    i\_d\_64 & 559 & 11 & 9 \\
   \hline
   \hline
    rand\_00100 & 100 & 5 & 4 \\
   \hline
    rand\_00200 & 200 & 8 & 6 \\
   \hline
    rand\_00400 & 400 & 12 & 9 \\
   \hline
    rand\_01000 & 1,000 & 17 & 14 \\
   \hline
    rand\_02000 & 2,000 & 28 & 21 \\
   \hline
    rand\_04000 & 4,000 & 37 & 30 \\
   \hline
    rand\_08000 & 8,000 & 47 & 38 \\
   \hline
    rand\_16000 & 16,000 & 76 & 61 \\
   \hline
    rand\_32000 & 32,000 & 95 & 74 \\
   \hline
    rand\_64000 & 64,000 & 137 & 110 \\
   \hline
    rand\_128000 & 128,000 & 206 & 164 \\
   \hline
\end{tabular}

    }
    \caption{Input used, and the algorithm performance on these
       inputs.}
    \figlab{results}
\end{figure}

\end{document}